\documentclass[12pt]{article}
\usepackage[utf8]{inputenc}
\usepackage{amsmath,amssymb}

\usepackage{etoolbox,fp}
\usepackage[dvipsnames,usenames]{xcolor}
\usepackage{tikz}

\newtheorem{theorem}{Theorem}[section]
\newtheorem{lemma}[theorem]{Lemma}
\newtheorem{proposition}[theorem]{Proposition}
\newtheorem{corollary}[theorem]{Corollary}
\newtheorem{example}[theorem]{Example}

\newcommand{\bull}{\mbox{$\;\;\;$\vrule height .9ex width .8ex depth -.1ex}}
\newenvironment{proof}{\par\smallbreak\noindent{\bf Proof.~}}
{\unskip\nobreak\hfill \bull \par\medbreak}
\newenvironment{proofof}[1]{\par\smallbreak\noindent{\bf Proof of~#1.~}}
{\unskip\nobreak\hfill \bull \par\medbreak}
\newcommand{\Case}[2]{\smallskip\par{\it Case #1:\/ #2}}
\newcounter{claim}
\renewcommand{\theclaim}{\Alph{claim}}
\newenvironment{claim}{\refstepcounter{claim}%
\par\medskip\par\noindent{\it Claim~\theclaim.~}~\rm}%
{\par\smallskip\par}

\newenvironment{bfenumerate}%
{%
\mbox{}\begin{enumerate}}{\end{enumerate}}

\newcommand{\of}[1]{\left( #1 \right)}
\newcommand{\Set}[1]{\left\{ #1 \right\}}
\newcommand{\refeq}[1]{(\ref{eq:#1})}
\newcommand{\setdef}[2]{\left\{ \hspace{0.5mm} #1 : \hspace{0.5mm} #2 \right\}}
\newcommand{\FO}[1]{\ensuremath{\mathrm{FO}}}
\newcommand{\tw}[1]{\mathit{tw}(#1)}
\newcommand{\wplanar}[1]{W(#1;\mathrm{planar})}
\newcommand{\csp}[1]{\ensuremath{\mathit{CSP}(#1)}}
\newcommand{\game}[3]{\ensuremath{\mathrm{COL}^{#1}_{#2}(#3)}}
\newcommand{\Tr}{T_{\mathrm r}}
\newcommand{\Tb}{T_{\mathrm b}}
\newcommand{\Tg}{T_{\mathrm g}}

\colorlet{acolor}{red}
\colorlet{bcolor}{blue}
\colorlet{ccolor}{ForestGreen}

\tikzset{
vertex/.style={circle,draw,inner sep=2pt,fill=none},
avertex/.style={circle,draw,inner sep=2pt,fill=acolor},
bvertex/.style={circle,draw,inner sep=2pt,fill=bcolor},
cvertex/.style={circle,draw,inner sep=2pt,fill=ccolor},
uvertex/.style={circle,draw=Black,inner sep=2pt,fill=white}
}

\title{On the dynamic width of the 3-colorability problem}
\author{Albert Atserias\\[2mm]
\small Universitat Politècnica de Catalunya
\and
Anuj Dawar\thanks{%
Supported by EPSRC grant EP/H026835.  Research visit to Germany supported by DAAD grant A/13/05456.}
\\[2mm]
\small University of Cambridge
\and
Oleg Verbitsky\thanks{%
Supported by DFG grant VE 652/1--1.
On leave from the Institute for Applied Problems of Mechanics and Mathematics,
Lviv, Ukraine.}\\[2mm]
\small Humboldt Universität zu Berlin}

\date{}

\begin{document}

\maketitle

\begin{abstract}
A graph $G$ is 3-colorable if and only if
it maps homomorphically to the complete 3-vertex graph $K_3$.
The last condition can be checked by a $k$-consistency
algorithm where the parameter $k$ has to be chosen large enough, dependent on $G$.
Let $W(G)$ denote the minimum $k$ sufficient for this purpose.
For a non-3-colorable graph $G$, $W(G)$ is equal to the minimum $k$ such that $G$
can be distinguished from $K_3$ in the $k$-variable
existential-positive first-order logic.
We define the \emph{dynamic width} of the 3-colorability 
problem as the function $W(n)=\max_G W(G)$,
where the maximum is taken over all non-3-colorable $G$
with $n$ vertices. 

The assumption $\mathrm{NP}\ne\mathrm{P}$ implies that $W(n)$ is unbounded.
Indeed, a lower bound $W(n)=\Omega(\log\log n/\log\log\log n)$
follows unconditionally from the work of Ne\v{s}et\v{r}il and Zhu \cite{NesetrilZ96}
on bounded treewidth duality.
The Exponential Time Hypothesis implies a much stronger bound $W(n)=\Omega(n/\log n)$ and
indeed we unconditionally prove that $W(n)=\Omega(n)$. In fact, an even stronger
statement is true: A first-order sentence distinguishing any 3-colorable graph on $n$ vertices
from any non-3-colorable graph on $n$ vertices must have $\Omega(n)$ variables.

On the other hand, we observe that $W(G)\le 3\,\alpha(G)+1$ and $W(G)\le n-\alpha(G)+1$
for every non-3-colorable graph $G$ with $n$ vertices, where
$\alpha(G)$ denotes the independence number of $G$.
This implies that $W(n)\le\frac34\,n+1$, 
improving on the trivial upper bound $W(n)\le n$.

We also show that $W(G)>\frac1{16}\, g(G)$ for every non-3-colorable graph $G$,
where $g(G)$ denotes the girth of~$G$.

Finally, we consider the function $W(n)$ over planar graphs and prove that
$W(n)=\Theta(\sqrt n)$ in the case.
\end{abstract}

\section{Introduction}

\paragraph{Consistency checking algorithms and the width of a CSP.}

If there is a homomorphism from a structure $A$
to a structure $B$, we will say that $A$  \emph{maps homomorphically to} $B$
and write $A\to B$. According to the framework of Feder and Vardi \cite{FederV98},
a (non-uniform) \emph{constraint satisfaction problem} is modelled  as the decision problem for 
the class of structures $\csp B=\setdef{A}{A\to B}$
determined by an appropriate template structure $B$.
Several archetypal NP-complete problems fit in with this setting.
For example, a graph $G$ is 3-colorable if and only if
$G\to K_3$, where $K_3$ denotes the complete graph on 3 vertices.
Thus, the 3-COLORABILITY problem is identical to $\csp{K_3}$.

$k$-Consistency algorithms are based on the concept of constraint propagation
and are used as a practical approach to the CSP since the seventies; see, e.g., 
the survey \cite{Bessiere06} and the recent complexity analysis in \cite{Berkholz13,BerkholzV13}.
The cases of $k=2,3$ are known as \emph{arc} and \emph{path} consistency,
respectively, and themselves have a large body of literature.
Instead of deciding if $A\to B$, a $k$-consistency algorithm checks 
if the pair $(A,B)$ fulfills a weaker combinatorial condition,
which we will call the \emph{$k$-consistency property}.\footnote{%
Saying that ``a pair $(A,B)$ has the $k$-consistency property'', we simplify
more customary expressions
like ``the strong $k$-consistency can be enforced or established on~$(A,B)$''.
We do not define this notion here; see \cite[Section 3]{AtseriasBD07} or equivalent combinatorial statements
in Sections \ref{s:duality} and~\ref{s:logic}.} 
This approach is incomplete in general. 
Always when $A\to B$, the pair $(A,B)$ has the $k$-consistency property for any $k$. 
However, if $A\not\to B$, the $k$-consistency property can still be satisfied if $k$ is too small.  
If $A\not\to B$, let $W_B(A)$ denote the minimum $k$ such that
$(A,B)$ has the $k$-consistency property.
We call this parameter \emph{$B$-width of~$A$}.
If $W_B(A)$ is bounded by a constant independent of $A$, the problem $\csp B$
(and the structure $B$) is said to have \emph{bounded width}.

The $k$-consistency property can be checked in time $(|A|+|B|)^{O(k)}$, where
$|A|$ and $|B|$ denote the number of elements in the structures.
It follows that all bounded width problems are solvable in polynomial time.
An algebraic characterization of bounded-width structures is given 
by Barto and Kozik~\cite{BartoK09}.

\paragraph{Graphs.}

If $H$ is a bipartite graph, $\csp H$ has width at most 3 (see, e.g., \cite{NesetrilZ96})
and, hence, is solvable in polynomial time. The dichotomy theorem of Hell
and Ne\v{s}et\v{r}il \cite{HellN90} says that $\csp H$ is NP-complete
whenever $H$ is non-bipartite. It follows that, for each non-bipartite $H$,
the values of $W_H(G)$ cannot be bounded by a constant unless $\mathrm{NP}=\mathrm{P}$.
This is proved by Ne\v{s}et\v{r}il and Zhu \cite{NesetrilZ96} unconditionally.
A straightforward consequence of their analysis is the existence of a graph $G$
on $n$ vertices, for any large enough $n$, such that $G\not\to H$ and
$W_H(G)=\Omega(\log\log n/\log\log\log n)$. We survey this approach in Section~\ref{s:duality}.

\paragraph{Dynamic width of $\csp H$.}

The main purpose of this paper is to pursue the 
analysis of the width parameter $W_H(G)$, focusing on its dynamic behavior
as a function of the number of vertices in $G$.
We define the \emph{dynamic width} of a graph $H$ (or of the corresponding
$\csp H$) as the function
$$
W_H(n)=\setdef{W_H(G)}{G\not\to H,\,|G|=n}.
$$
Note that $W_H(n)$ is well defined for all sufficiently large $n$,
at least for all $n$ greater than the chromatic number of $H$.
As it was already mentioned, $W_H(n)\le3$ for any bipartite $H$.
If $H$ is not bipartite, then from the work of Ne\v{s}et\v{r}il and Zhu \cite{NesetrilZ96} 
it follows that 
\begin{equation}
  \label{eq:NesZhu}
  W_H(n)=\Omega\of{\frac{\log\log n}{\log\log\log n}}.
\end{equation}

To emphasize the importance of this notion, suppose that we know that
$W_H(n)=O(k(n))$ where $k(n)$ is a function computable in time
bounded by a polynomial\footnote{%
If $W_H(n)$ is unbounded, this can be relaxed to the computability
in time $n^{O(k(n)}$.}
in $n$.
In this case $\csp H$ is solvable in time $2^{O(k(n)\log n)}$,
which can prove to be a nontrivial algorithmic result even when this time bound
is superpolynomial.

\paragraph{The 3-colorability problem.}

In this paper, we focus our attention on the dynamic width of $\csp{K_3}$, that is,
the 3-COLORABILITY problem. To facilitate the notation, let $W(G)=W_{K_3}(G)$
and $W(n)=W_{K_3}(n)$. Dawar \cite{Dawar98} proves that 3-COLORABILITY
is not definable in the infinitary logic with finitely many variables.
The argument of \cite{Dawar98} immediately translates to the bound $W(n)=\Omega(\sqrt n)$.

The currently best algorithm for 3-COLORABILITY \cite{BeigelE05} runs in time $O(1.3289^n)$.
 It is known \cite{ImpagliazzoPZ01} that 3-COLORABILITY is not solvable in time
$2^{o(n)}$ unless the \emph{Exponential Time Hypothesis} fails.
Therefore, under this hypothesis we should have
a lower bound at least as strong as $W(n)=\Omega(n/\log n)$.
Our main results is an unconditional linear lower bound, i.e.,
\begin{equation}
  \label{eq:f-lower}
W(n)=\Omega(n).  
\end{equation}
The proof of \refeq{f-lower} is based on the logical characterization
of the parameter $W(G)$ (see Section \ref{s:logic}) and exploits the same method as used in~\cite{Dawar98}.

A straightforward observation $W_H(G)\le|G|$ implies that $W(n)\le n$.
By \refeq{f-lower}, this trivial upper bound can be improved at most up to a
constant factor. We show that such an improvement is really possible,
noticing that
\begin{equation}
  \label{eq:f-upper}
W(n)\le\frac34n+1.  
\end{equation}
This bound follows from the relation
$$W(G)\le\min\{3\,\alpha(G),n-\alpha(G)\}+1,$$ 
where $\alpha(G)$ denotes the independence number of~$G$.

We also relate the parameter $W(G)$ to the girth $g(G)$ of the graph, proving that 
$$W(G)>\frac1{16}\, g(G)$$
(cf.\ the bound \refeq{loglogg} below).
This relation implies only logarithmic lower bound for $W(n)$
but has an advantage of being true for every non-3-colorable graph~$G$.

Finally, we consider the function $W(n)$ over planar graphs and prove that
$W(n)=\Theta(\sqrt n)$ in the case.

\paragraph{Related work: The dynamic width of other CSPs.}

As it is well known, the 3-SATISFIABILITY problem can be encoded
as a CSP. Atserias \cite[Theorem 2]{Atserias04} obtains a result that,
using our notation, can be stated as
\begin{equation}
  \label{eq:3sat}
  W_{\mathrm{3SAT}}(n)=\Omega\of{\frac n{\log^2n}},
\end{equation}
where $n$ is the number of propositional variables in an input 3CNF.
Moreover, \refeq{3sat} is shown for 3SAT instances with $O(n)$ clauses
and is an average-case rather than
a worst-case bound. Note that \refeq{3sat} is close to
the bound $W_{\mathrm{3SAT}}(n)=\Omega\of{\frac n{\log n}}$
that follows from the Exponential Time Hypothesis~\cite{ImpagliazzoPZ01}.
Atserias, Bulatov and Dawar \cite{AtseriasBD09} prove that testing the solvability
of systems of equations over a finite Abelian group and related CSPs
are not definable in the infinitary logic with finitely many variables
(even with counting), which implies that the dynamic width of these problems
is unbounded.

\section{Treewidth duality}\label{s:duality}

\paragraph{Graph-theoretic preliminaries.}

An {\em $s$-coloring\/} of a graph $G$ is a map from the set of vertices $V(G)$ to 
the set of colors $\{1,2,\ldots,s\}$.
A coloring $c$ is {\em proper\/} if $c(u)\ne c(v)$ for any adjacent vertices
$u$ and $v$. A graph $G$ is {\em $s$-colorable\/} if it has a proper
$s$-coloring. The minimum $s$ for which $G$ is $s$-colorable is called
the {\em chromatic number\/} of $G$ and denoted by $\chi(G)$. If
$\chi(G)=s$, then $G$ is called {\em $s$-chromatic}.
A set of vertices is \emph{independent} if all of them are pairwise non-adjacent.
The \emph{independence number} $\alpha(G)$ of a graph $G$ is the maximum size
of an independent set in $G$.
In a proper coloring of $G$, any set of vertices with the same color is independent.
This implies that $\chi(G)\alpha(G)\ge n$, where $n$ denotes the number
of vertices in~$G$.

The \emph{girth} $g(G)$ of a graph $G$ is the minimum length
of a cycle in~$G$.

\begin{proposition}{\bf (Erd\H{o}s \cite{Erdos59})\ }\label{prop:erdos}
For every $n\ge s$ there is an $s$-chromatic graph $G$ on $n$ vertices with 
$$
g(G)=\Omega\of{\frac{\log n}{\log s}}.  
$$
\end{proposition}

The logarithmic bound in Proposition \ref{prop:erdos} is best possible.
In \cite{Erdos62}, Erd\H{o}s proves that the girth of an $s$-chromatic graph
on $n$ vertices is bounded by $\frac{2\log n}{\log(s-2)}+1$.

\paragraph{The Ne\v{s}et\v{r}il-Zhu bound.}

A useful combinatorial bound for the $H$-width of $G$ is due to Freuder \cite{Freuder90}
who showed that if $G\not\to H$, then 
\begin{equation}
  \label{eq:tw}
W_H(G)\le\tw G+1,  
\end{equation}
where $\tw G$ denotes the treewidth of $G$.\footnote{%
All what is stated here in the language of graphs holds true
also for general relational structures. The treewidth of structure $A$
is defined as the treewidth of its Gaifman graph.}
A complete combinatorial characterization of the $H$-width
is suggested by Hell, Ne\v{s}et\v{r}il, and Zhu \cite{HellNZ96}
for digraphs and extended to general structures by Feder and Vardi \cite{FederV98}. 
Note that $F\to G$ and $F\not\to H$
imply $G\not\to H$. In view of this, we call such a graph $F$
an \emph{$H$-obstruction for $G$}. 
It turns out that $W_H(G)$ is equal to the minimum $k$ such that
$G$ has an $H$-obstruction of treewidth $k-1$.
Note that \refeq{tw} follows from here because if $G\not\to H$,
then $G$ is an $H$-obstruction for itself.

Thus, the statement that the width of $H$ is bounded by $k$
can equivalently be expressed in the following form:
$G\not\to H$ if and only if $F\to G$ for some $F$ such that $F\not\to H$ and $\tw F<k$.
This homomorphism duality justifies the terminology, according to which
graphs $H$ of width $k$ are also said to have \emph{treewidth-$(k-1)$ duality};
see the survey \cite{BulatovKL08}.


Ne\v{s}et\v{r}il and Zhu \cite{NesetrilZ96} show that no non-bipartite
graph $H$ has bounded treewidth duality. This is a consequence of Proposition \ref{prop:erdos}
and the following fact established in \cite{NesetrilZ96}:
If 
\begin{equation}\label{eq:gkk}
g(G)>2^{k+2}(4km)^{4km-1}+2(k+1), 
\end{equation}
$\tw F\le k$, and $F\to G$, then
also $F\to C_{2m+1}$, where $C_{2m+1}$ denotes the cycle of length $2m+1$.
Indeed, suppose that $H$ contains $C_{2m+1}$ as a subgraph. 
Consider a graph $G$ such that $G\not\to H$.
Assuming that $g(G)>8(4m)^{4m-1}$, set $k$ to the largest value
such that \refeq{gkk} is fulfilled; note that $k=\Omega(\log g(G)/\log\log g(G))$.
Then no $F$ of treewidth at most $k$
can serve as an $H$-obstruction for $G$.
It follows that
\begin{equation}
  \label{eq:loglogg}
  W_H(G)=\Omega\of{\frac{\log g(G)}{\log\log g(G)}}
\end{equation}
(due to using the $\Omega$-notation, this bound is trivially true also for $G$
with $g(G)\le8(4m)^{4m-1}$).
The bound \refeq{NesZhu} we stated above can be derived from \refeq{loglogg}
and Proposition \ref{prop:erdos}. The latter gives us a graph $G$
with logarithmic girth and $\chi(G)>\chi(H)$.
The last condition ensures that $G\not\to H$.

\section{Existential-positive $k$-variable logic and existential $k$-pebble game}\label{s:logic}

For graphs we consider the first-order language with two relation symbols
for vertex adjacency and equality.
An \emph{existential-positive first-order} formula $\Phi$ is built using only
monotone Boolean connectives (i.e., conjunction and disjunction)
and existential quantification. 
If such a $\Phi$ is true on a graph $G$
and $G\to H$, then $\Phi$ must hold true also on the graph $H$. 
Moreover, for every finite $G$
there is an existential-positive statement $\Phi_G$ that is true on $H$
if and only if $G\to H$. To obtain $\Phi_G$, we can assign a first-order
variable to each vertex of $G$ and list all adjacency relations between the variables
that are true in $G$ for the corresponding vertices.
This gives us a logical characterization of the homomorphism relation:
$G\not\to H$ if and only if there is an existential-positive statement $\Phi$ that
\emph{distinguishes} $G$ from $H$, that is, $\Phi$ is true on $G$ but false on~$H$.

We define the \emph{width} $W(\Phi)$ of a first-order formula $\Phi$
to be the number of variables occurring in it; different occurrences of
the same variable do not count. The \emph{$k$-variable logic}
consists of formulas of width at most $k$. Suppose that $G\not\to H$.
Kolaitis and Vardi \cite{KolaitisV95} show that $W_H(G)$ is equal to the minimum $k$
such that $G$ is distinguishable from $H$ in existential-positive
$k$-variable logic.

The logical characterization of $W_H(G)$ implies also a useful
combinatorial characterization of this parameter \cite{KolaitisV95}.
The \emph{existential $k$-pebble game} on graphs $G$ and $H$
is a version of the $k$-pebble Ehrenfeucht-Fra\"iss\'e game where
Spoiler always moves in $G$ and
Duplicator's objective is to keep a partial homomorphism.
The parameter $W_H(G)$ is equal to the minimum $k$
such that Spoiler has a winning strategy in the game.
We use this characterization in Sections \ref{s:ind} and~\ref{s:girth}.

\section{A linear lower bound for $W(n)$}\label{s:lower}

Assume that $\chi(G)>3$ and
consider an arbitrary graph $H$ with $\chi(H)\le3$.
Duplicator can use a homomorphism from $H$ to $K_3$
to translate her strategy in the $k$-pebble existential game on $G$ and $H$
into a strategy in the game on $G$ and $K_3$. Therefore,
$W_H(G)\le W(G)$. It follows that
$$
W(n)=\max_{G,H}\setdef{W_H(G)}{|G|=n,\ \chi(G)>3\text{ and }\chi(H)\le3}.
$$
Let $W^*(G,H)$ denote the minimum width of a first-order statement (with no restrictions) distinguishing $G$
from $H$. Obviously, $W^*(G,H)\le W_H(G)$.
Define
$$
W^*(n)=\max_{G,H}\setdef{W^*(G,H)}{|G|=|H|=n,\ \chi(G)>3\text{ and }\chi(H)\le3}
$$
and note that 
$$
W^*(n)\le W(n).
$$
Thus, in order to estimate $W(n)$ from below, it suffices
to prove a lower bound for $W^*(n)$. For this,
we will use the approach of \cite{Dawar98}, which in turn
is based on the Cai-Fürer-Immerman \cite{CaiFI92} construction
of non-isomorphic graphs $G$ and $H$ on $n$ vertices
that cannot be distinguished in first-order logic 
with bounded number of variables (even when counting quantifiers are allowed).
In \cite{Dawar98}, this construction is enhanced to ensure that
one of the graphs $G$ and $H$ is 3-colorable and the other is not.

We will need the following notions.
Recall that a graph is \emph{uniquely} 3-colorable if it is 3-colorable and the
coloring is unique up to a renaming of colors.
Let $X\subseteq V(G)$. The result of removal of all vertices in $X$ from $G$ is denoted by
$G-X$. We call $X$ a \emph{separator} of $G$ if every connected component of $G-X$ has at most
$|V(G)|/2$ vertices.

\begin{lemma}{\bf (\cite{Dawar98})}\label{lem:indist}
  Suppose that $A$ is a graph with the following properties:
  \begin{itemize}
  \item 
 $A$ has $m$ vertices and maximum degree $d$;
\item 
  $A$ is uniquely 3-colorable;
\item 
  $A$ has no separator with $k$ vertices.
  \end{itemize}
Then, one can construct from $A$ two graphs $G_A$ and $H_A$ with the same number
$n=O(dm)$ of vertices and $O(n)$ edges so that $\chi(G_A)\le3$, $\chi(H_A)>3$, and $W^*(G_A,H_A)>k$.
\end{lemma}

We are now ready to prove our linear lower bound. 

\begin{theorem}\label{thm:lower}
$W^*(n)=\Omega(n)$.
\end{theorem}

\begin{proof}
It suffices to find a graph $A$ with the properties listed in Lemma \ref{lem:indist},
for an arbitrarily large $m$, constant $d$, and linear $k=\Omega(m)$.
This can be done in two steps.

First, take a connected $d'$-regular graph $B$ 
with $m'$ vertices such that $\chi(B)>3$ and the minimum separator size in $B$
is larger than $k=\Omega(m')$. Specifically, we can set $d'=6$.
A random 6-regular graph has the required properties. With high probability,
it is connected and has only linear separators (Pinsker \cite{Pinsker73}), and
its chromatic number is 4 (Shi and Wormald \cite{ShiW07}).

Next, we take $A=B\times K_3$, where $\times$ denotes the categorical product
of graphs. That is,
each vertex $v$ of $B$ has three copies in $A$, namely $(v,1)$, $(v,2)$, and $(v,3)$,
and vertices $(v,i)$ and $(u,j)$ are adjacent in $A$ if
$v$ and $u$ are adjacent in $B$ and $i\ne j$ (i.e., $i$ and $j$ are adjacent
in the complete graph $K_3$ on the vertex set $\{1,2,3\}$).

The graph $A$ is 3-colorable as each of the three copies of $V(B)$ in $A$ are
independent sets. This coloring is unique by the Greenwell-Lov\'asz theorem \cite{GreenwellL74},
which says that the categorical product $G\times K_s$ is uniquely $s$-colorable
whenever $G$ is connected and $\chi(G)>s$.

The graph $A$ has $m=3m'$ vertices, and all vertices of $A$ have degree 12.
Note that $k=\Omega(m')=\Omega(m)$.
It remains to show that, like $B$, the graph $A$ has no separator of size~$k$.

To this end, consider a set $X$ of vertices in $A$ such that $|X|\le k$.
Let $X'$ be the projection of $X$ onto the first coordinate, that is,
the set of those vertices in $B$ that occur as the first
components of the vertices in $X$. Obviosly, $|X'|\le|X|\le k$. 
Therefore, $B-X'$ has a connected
component $C'$ of size exceeding $m'/2$. Let us lift $C'$ to $A$ and denote 
the resulting set by $C$; that is,
let each vertex $v$ in $C'$ contribute three vertices $(v,1)$, $(v,2)$, and $(v,3)$ into $C$.
Since $C'$ spans a connected subgraph in $B$, the set $C$ spans a connected subgraph
in $A$ (because if $v$ and $u$ are adjacent vertices in $C'$ then their clons
$(v,1)$, $(v,2)$, $(v,3)$, $(u,1)$, $(u,2)$, and $(u,3)$ in $C$ span
a connected subgraph $K_2 \times K_3=C_6$ in $A$). Note now that $C$ and $X$ are disjoint by construction.
Therefore, $C$ is a connected component of $A-X$ having size $|C|=3|C'|>3m'/2=m/2$,
and $X$ cannot be a separator of~$A$.

Theorem \ref{thm:lower} follows now by Lemma~\ref{lem:indist}.
\end{proof}

\section{Relationship between the width and the independence number}\label{s:ind}

The main technical tool in our further analysis is the existential $k$-pebble game on 
graphs $G$ and $H$.
For the special case that $H=K_3$, we recast it
in slightly different terms. The \emph{$k$-width 3-coloring game on a graph $G$}
is played by Spoiler and Duplicator. In a round of the game Spoiler
selects a vertex in $G$ and then Duplicator colors it in one of three colors,
red, blue, or green. After each round, at most $k$ vertices are allowed
to be colored. To obey this condition, Spoiler can erase the color
of a previously colored vertex before he demands to color a new one.
Duplicator wins the $r$-round game if during the play there is no
two adjacent vertices colored in the same color
(i.e., after each round the partial 3-coloring of $G$ is proper).
The following fact is a particular case of the relationship 
between the existential $k$-pebble game and existential $k$-variable logic~\cite{KolaitisV95}.

\begin{proposition}\label{prop:game}
Suppose that $\chi(G)>3$. Then
$W(G)$ is equal to the minimum $k$ such that, for some $r$, Spoiler has a winning strategy
in the $r$-round $k$-width 3-coloring game on~$G$.  
\end{proposition}

We now relate the width of a graph to its independence number.

\begin{theorem}\label{thm:ind}
Let $G$ be a graph with $n$ vertices. If $\chi(G)>3$, then
$$W(G)\le\min\{3\,\alpha(G),n-\alpha(G)\}+1.$$
\end{theorem}

\begin{proof}
Let $v(H)$ denote the number of vertices in a graph $H$;
thus, $v(G)=n$.
  We first prove the bound 
  \begin{equation}
    \label{eq:alpha}
W(G)\le v(G)-\alpha(G)+1.   
  \end{equation}
Let $U$ be an independent set in $G$ with $\alpha(G)$ vertices.
Consider the 3-coloring game on $G$ and let Spoiler
select all vertices in $V(G)\setminus U$.
Suppose that Duplicator manages to properly color this subgraph of $G$.
This partial coloring does not extend properly to
some vertex $u\in U$ for else the whole $G$ would be 3-colorable.
In the next round Spoiler selects also this vertex and wins.

To prove the bound $W(G)\le 3\,\alpha(G)+1$,
we apply the bound \refeq{alpha} to a smallest 4-chromatic
induced subgraph $G'$ of $G$. Note that $v(G')\le4\,\alpha(G')$
and $\alpha(G')\le\alpha(G)$. Therefore,
$$
W(G)\le W(G')\le v(G')-\alpha(G')+1\le 3\,\alpha(G')+1\le 3\,\alpha(G)+1,
$$
as claimed.
\end{proof}

Theorem \ref{thm:ind} immediately implies an improvement on the trivial upper bound $W(n)\le n$.

\begin{corollary}\label{cor:upper}
 $W(n)\le\frac34n+1$.
\end{corollary}

\section{Relationship between the width and the girth}
\label{s:girth}

We here show a relation between the width of a non-3-colorable graph $G$,
which is much stroger than the bound \refeq{loglogg} in the case $H=K_3$.

\begin{theorem}\label{thm:girth}
If $\chi(G)>3$, then
$W(G)>\frac1{16}\,g(G)$.
\end{theorem}

The proof of this result takes the rest of this section.
It is based on Proposition \ref{prop:game}. We will show that if Spoiler
has a winning strategy in the $r$-round $k$-width 3-coloring game on $G$,
then $r>\lfloor\log_4(g(G)-2)\rfloor$ and $k>\frac1{16}\,g(G)$.

We write $\game krH$ to denote the $r$-round $k$-width 3-coloring game on 
a graph $H$, where up to $k$ vertices of $H$ can be precolored
before the game begins. Such graphs will be called \emph{$k$-precolored}.
A subgraph $H'$ of a $k$-precolored graph $H$ can inherit vertex colors
present in $H$. That is, if a vertex is colored in $H'$, it must have
the same color in $H$; on the other hand, a vertex colored in $H$ can 
be uncolored in $H'$. Note that a \emph{proper subgraph} $H'$ of $H$
can have all the vertices and the edges as in $H$, but then 
there must be a vertex colored in $H$ and uncolored in $H'$.
If $H'$ and $H$ are $k$-precolored graphs,
then a homomorphism from $H'$ to $H$ has to preserve colors of vertices
in $H'$ (but it can map an uncolored vertex of $H'$ to a colored vertex of $H$).
Saying that a player \emph{wins} a game, we mean that s/he has a winning strategy
in it.

\begin{lemma}\label{lem:play}

\begin{bfenumerate}
\item 
If Spoiler wins $\game kr{H'}$ and $H'$ is a subgraph of a $k$-precolored graph $H$,
then he wins also $\game kr{H}$.
\item 
If Duplicator wins $\game kr{H}$ and $H'\to H$,
then she wins also $\game kr{H'}$.
\end{bfenumerate}
\end{lemma}

\begin{proof}
  The first part is straightforward and also formally follows from the second. 
To prove the second part,
suppose that $h$ is a homomorphism from $H'$ to $H$.
Duplicator simulates $\game kr{H'}$ by $\game krH$.
When Spoiler asks for color of a vertex $v$ in $H'$,
Duplicator interprets this as the request for coloring $h(v)$ in $H$
and colors $v$ according to her strategy in $\game krH$.
This strategy is winning for her in $\game kr{H'}$ because $h$ stays
to be a homomorphism in each round.
\end{proof}

Lemma \ref{lem:play}.1 motivates the following definition.
Define $M^k_r$ to be the family of all minimal $k$-precolored graphs $H$
such that Spoiler wins $\game krH$, where minimality means that
Duplicator wins $\game krK$ for any proper subgraph $K$ of~$H$.

\begin{example}\label{ex:}\rm\mbox{}\\[2mm]
  $\begin{array}{rclcl}
M^k_0&=&\emptyset&\text{if}&k=0,1;\\[1mm]
M^k_0&=&\Set{
\begin{tikzpicture}[scale=.5]
\path (0,0) node[avertex] (a) {}
      (1,0) node[avertex] (b)  {} edge (a);
\end{tikzpicture},
\begin{tikzpicture}[scale=.5]
\path (0,0) node[bvertex] (a) {}
      (1,0) node[bvertex] (b)  {} edge (a);
\end{tikzpicture},
\begin{tikzpicture}[scale=.5]
\path (0,0) node[cvertex] (a) {}
      (1,0) node[cvertex] (b)  {} edge (a);
\end{tikzpicture}
}&\text{if}&k\ge2;\\[1.5mm]
M^k_r&=&M^k_0&\text{if}&k=2,3;\\
M^k_1&=&M^k_0\cup\Set{
\raisebox{-3mm}{
\begin{tikzpicture}[scale=.4]
\path (-.7,-.7) node[avertex] (a) {}
      (0,0) node[vertex] (b)  {} edge (a)
      (.7,-.7) node[bvertex] (c)  {} edge (b)
      (0,1) node[cvertex] (d)  {} edge (b);
\end{tikzpicture}
}}&\text{if}&k\ge4.
  \end{array}$\\[1.5mm]
A few other examples are shown in Fig.~\ref{fig:}.
\end{example}

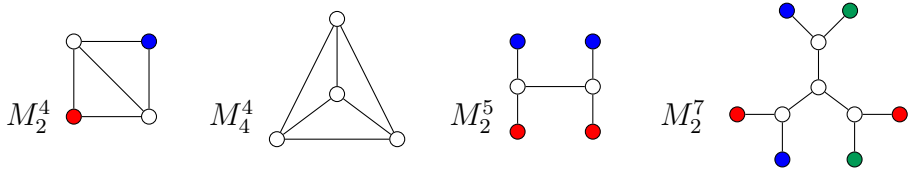
\begin{figure}
\centering
 \begin{tikzpicture}
\begin{scope}
\node at (-.6cm,0cm) {$M^4_2$};
\path (0,0) node[avertex] (a) {}
      (1,0) node[vertex] (b)  {} edge (a)
      (0,1) node[vertex] (c)  {} edge (a) edge (b)
      (1,1) node[bvertex] (d)  {} edge (b) edge (c);
\end{scope}

\begin{scope}[xshift=3.5cm]
\node at (-1.4cm,0cm) {$M^4_4$};
\path[every node/.style=vertex,yshift=3mm] 
      (-.8,-.6) node (a) {}
      (0,0) node (b)  {} edge (a)
      (.8,-.6) node (c)  {} edge (a) edge (b)
      (0,1) node (d)  {} edge (a) edge (b) edge (c);
\end{scope}

\begin{scope}[xshift=5.9cm]
\node at (-.6cm,0cm) {$M^5_2$};
\path[yshift=4mm]
      (0,0) node[vertex] (a) {}
      (0,-.6) node[avertex] (b)  {} edge (a)
      (0,.6) node[bvertex] (c)  {} edge (a)
      (1,0) node[vertex] (d) {} edge (a)
      (1,-.6) node[avertex] (e)  {} edge (d)
      (1,.6) node[bvertex] (f)  {} edge (d);
\end{scope}

\begin{scope}[xshift=9.9cm]
\node at (-1.8cm,0cm) {$M^7_2$};
\path[scale=.6,yshift=6.5mm] 
      (-.8,-.6) node[vertex] (a) {}
      (0,0) node[vertex] (d)  {} edge (a)
      (.8,-.6) node[vertex] (b)  {} edge (d)
      (0,1) node[vertex] (c)  {} edge (d)
      (-1.8,-.6) node[avertex] (a1) {}  edge (a)
      (-.8,-1.6) node[bvertex] (a2) {}  edge (a)
      (1.8,-.6) node[avertex] (b1) {}  edge (b)
      (.8,-1.6) node[cvertex] (b2) {}  edge (b)
      (-.7,1.7) node[bvertex] (c1) {}  edge (c)
      (.7,1.7) node[cvertex] (c2) {}  edge (c);
\end{scope}
 \end{tikzpicture}
\caption{Representatives of some $M^k_r$ families.}
\label{fig:}
\end{figure}

\begin{lemma}\label{lem:diam}
Every graph in $M^k_r$ has diameter at most~$2^r$.
\end{lemma}

\begin{proof}
  Denote the maximum diameter of a graph in $M^k_r$ by $D^k_r$.
Let $H\in M^k_r$. Fix a winning strategy for Spoiler in $\game krH$.
Let $v$ be the vertex claimed by Spoiler in the first round
(possibly after erasing the color of some vertex).
Duplicator can color $v$ in one of the three colors;
denote the three possible results by $H_1,H_2,H_3$
(that is, those are $k$-precolored graphs which can appear after playing the first round).
Spoiler wins the game $\game k{r-1}{H_i}$ for every $i=1,2,3$.
Therefore, each $H_i$ contains a subgraph $H'_i$ belonging to $M^k_{r-1}$.
Each $H'_i$ contains the vertex $v$; otherwise $H'_i$ would be
a proper subgraph $K$ of $H$ such that Spoiler wins $\game krK$.
Furthermore, note that $\bigcup_{i=1}^3V(H'_i)=V(H)$; otherwise the union
of $H'_1$, $H'_2$, and $H'_3$ with the vertex $v$ uncolored
would be a proper subgraph $K$ of $H$ such that Spoiler wins $\game krK$.
It follows that every vertex in $H$ is reachable from $v$
within distance $D^k_{r-1}$ and, therefore, $D^k_r\le2D^k_{r-1}$.
Note that $D^k_0=1$; see Example \ref{ex:}. By induction, $D^k_r\le2^r$.
\end{proof}

Let $M^k_r(G)$ consist of those $k$-precolored graphs in $M^k_r$
whose underlying (uncolored) graphs appear as subgraphs in~$G$.

\begin{lemma}\label{lem:is-tree}
  If $r\le\log_2(g(G)-2)-1$, then every graph in $M^k_r(G)$
is a ($k$-precolored) tree.
\end{lemma}

\begin{proof}
Let $H\in M^k_r(G)$ where $r\le\log_2(g(G)-2)-1$. 
It follows from Lemma \ref{lem:diam} that the diameter of $H$ is smaller
than $(g(G)-1)/2$ (in particular, $H$ is connected). 
A subgraph of $G$ with such diameter must be acyclic
because any cycle in $G$ contains two vertices at the distance
at least $(g(G)-1)/2$ from each other.
\end{proof}

Fix $R=\lfloor\log_2(g(G)-2)\rfloor-1$.
Since any tree in $M^k_r$ must have colored vertices,
Lemma \ref{lem:is-tree} implies that
$G$ does not have any subgraph in $M^k_R$.
By the definition of $M^k_r$,
Spoiler cannot win $\game kRG$, whatever~$k$.

Our further proof strategy is the following.
Fix $k$ such that Spoiler wins $\game ksG$ for some $s>R$.
Let $\kappa_r$ denote the minimum number of colored vertices
in a tree from $M^k_r\setminus M^k_{r-1}$.
Assume that both Spoiler and Duplicator play optimally, that is,
Spoiler always minimizes the number of rounds till his win, while Duplicator
always maximizes the number of rounds till her loss.
In the round that Spoiler is able to win within the next $R$ moves, the graph
$G$ (endowed with the current partial coloring) 
must contain a subgraph $S$ belonging to $M^k_R$
and cannot contain any subgraph from $M^k_{R-1}$.
By Lemma \ref{lem:is-tree}, the subgraph $S$ is a tree from $M^k_R\setminus M^k_{R-1}$
and, therefore, $k\ge\kappa_R$. In the remainder of the proof
we will estimate the value of $\kappa_R$ from below.

Given a tree $T$ and its vertex $v$, let $T-v$ denote the graph
obtained from $T$ by removing $v$. A \emph{$v$-branch} of $T$
is a subtree of $T$ containing a connected component of $T-v$
along with the vertex~$v$.

\begin{lemma}\label{lem:sep}
  Let $T$ be a tree with $n\le2^r+2$ vertices and $l<k$ leaves. Suppose that
the leaves of $T$ are colored and that this coloring does not extend
to a proper 3-coloring. Then Spoiler wins $\game krT$.
\end{lemma}

\begin{proof}
  We proceed by induction on $r$. For the base cases of $r=0,1$, see the entries for
$M^k_0$ and $M^k_1$ in Example \ref{ex:}. If $n>4$, Spoiler uses a standard separator strategy.
Every tree $T$ with $n$ vertices has a single-vertex separator, that is, a vertex $v$
such that every component of $T-v$ has at most $n/2$ vertices.
In the first round Spoiler asks Duplicator to color such $v$.
Whatever color is used by Duplicator, there is a $v$-branch whose coloring is not properly 
extendable. From now on Spoiler plays in this branch and wins in the remaining $r-1$
rounds by the induction assumption because every $v$-branch has at most $2^{r-1}+2$ vertices.
\end{proof}

\begin{lemma}\label{lem:trees-M}
Every tree $T$ in $M^k_r$ has the following properties.
  \begin{bfenumerate}
  \item 
All leaves and no other vertices in $T$ are colored.
\item 
Every non-leaf vertex in $T$ has degree~3.
\item 
$T$ has less than $k$ leaves.
  \end{bfenumerate}
\end{lemma}

\begin{proof}
 {\bf 1.} 
We first argue that all leaves are colored.
Suppose that $T$ has at least 3 vertices; otherwise the claim is trivial.
Assume that $T$ has a non-colored leaf $v$ adjacent to a vertex $u$.
Mapping $v$ to another neighbor of $u$ and each other vertex to itself 
is a homomorphism of $T$ onto $T-v$. By the definition of $M^k_r$,
Duplicator has a winning strategy in $\game kr{T-v}$.
By Lemma \ref{lem:play}.2, she can win also $\game kr{T}$, a contradiction.

Assume now that, besides all leaves, a  non-leaf $w$ is also colored.
Then there is a $w$-branch $B$ such that Spoiler wins $\game krB$,
which contradicts the condition that $T\in M^k_r$.
(If Duplicator could win $\game krB$ for all $w$-branches $B$,
the precoloring of $T$ would extend to a proper 3-coloring by Lemma \ref{lem:sep};
therefore, Duplicator could win $\game kr{T}$ as well).

 {\bf 2.} 
We first prove that all non-leaf vertices in $T$ 
have degree no more than 3.
We use induction on the number of vertices in $T$.
The base case, when $T$ has at most 4 vertices, is straightforward;
see Example \ref{ex:}. If $T$ has more vertices, consider the game
$\game krT$. 
Fix a winning strategy for Spoiler. We can assume that
Spoiler never asks Duplicator to recolor a colored vertex or a vertex
whose color is just erased.
Let $v$ be the first move by Spoiler.
By Part 1, $v$ is a non-leaf.
Duplicator can color $v$ in 3 ways.
In each case Spoiler wins in less than $r$ moves and, hence,
the tree $T$ with the vertex $v$ colored contains a subtree,
$\Tr$, $\Tb$, or $\Tg$ (analogous to $H'_1,H'_2,H'_3$ in the proof of
Lemma \ref{lem:diam}), belonging to $M^k_{r-1}$.
Like in the proof of Lemma \ref{lem:diam}, the definition of $M^k_r$
implies that $\Tr$, $\Tb$, and $\Tg$ share the vertex $v$ and cover the whole tree $T$.
By Part 1, each of $\Tr$, $\Tb$, and $\Tg$ is a $v$-branch of $T$.
By the induction assumption, all non-leaf
vertices in each of the branches have degree at most 3.
Therefore, all non-leaf vertices in $T$, including $v$, have degree at most~3.

It remains to argue that no non-leaf vertex of $T$ has degree 2.
Using induction on the number of vertices, we show that if $T$
has a vertex of degree 2, then any coloring of the leaves admits
a proper extension to the whole tree. If $T$ has at most 4 vertices,
this is true by straightforward inspection because $T$ is then a path on 3 or 4 vertices. 
Suppose that $T$ has more than 4 vertices.
Let $v$ be a vertex of degree 2. The graph $T-v$ consists of two parts,
$T_1$ and $T_2$. Let $v_i$ be the vertex of degree 3 or 1 in $T_i$ nearest to $v$.
Denote the path from $v_1$ to $v_2$ by $P$. Remove all intermediate
vertices of $P$ from $T$. This splits $T$ into two parts $T'_1$ and $T'_2$,
containing $v_1$ and $v_2$ respectively. If $v_i$ has degree 1 in $T$, then $T'_i$
has no other vertex.
If $v_i$ has degree 3 in $T$, then it has degree 2 in $T'_i$.
By the induction assumption, the coloring of each $T'_i$ extends to a proper
coloring of this part. Finally, these proper colorings of $T'_1$ and $T'_2$
extend to a proper coloring of $T$ along~$P$.

 {\bf 3.} 
Part 1 readily implies that $T$ has at most $k$ leaves, and we have to show
that it cannot have precisely $k$ leaves. Consider $\game krT$ and let
the vertex $v$ and the $v$-branches $\Tr$, $\Tb$, and $\Tg$ be as in the proof
of Part 2. That is, $\Tr$, $\Tb$, and $\Tg$ are minimal subtrees
such that Spoiler wins each of the games
$\game k{r-1}{\Tr}$, $\game k{r-1}{\Tb}$, and $\game k{r-1}{\Tg}$.
Recall that these subtrees belong to $M^k_{r-1}$ and cover~$T$.

Now, suppose that $T$ has $k$ leaves. To be able to color $v$, Spoiler 
must erase the color of one of the leaves. 
Since this leaf belongs to one of $\Tr$, $\Tb$, and $\Tg$,
we obtain a contradiction because, by Part 1, no member of $M^k_{r-1}$
can have an uncolored leaf.
\end{proof}

We are now ready to finish the proof of Theorem~\ref{thm:girth}.
Consider a tree $T$ in $M^k_R\setminus M^k_{R-1}$. Suppose that $T$ has $n$ vertices,
$l$ of which are leaves. 
Lemma \ref{lem:trees-M}.3 implies that $l<k$.
Under this condition, Lemma \ref{lem:sep} implies that $n>2^{R-1}+2$.
From Lemma \ref{lem:trees-M}.2 it follows that $l=n/2+1$. Therefore,
$$
l>2^{R-2}+2=2^{\lfloor\log_2(g(G)-2)\rfloor-3}+2>\frac1{16}\,g(G).
$$
By Lemma \ref{lem:trees-M}.1, $T$ has more than $g(G)/16$ colored vertices.
We conclude that $\kappa_R>g(G)/16$, thereby completing
the proof.

\section{Planar 3-colorability}

The 3-colorability problem has been actively studied also
for particular classes of graphs.
Estimation of $W(n)$ would be also meaningful for
such classes. Specifically, we define the dynamic width function over a class of graphs $\mathcal C$
by $W(n;\mathcal C)=\max\setdef{W(G)}{G\in{\mathcal C},\,|G|=n,\,\chi(G)>3}$.
Here we consider the dynamic width $\wplanar n$ for the class of planar graphs.

Though 3-COLORABILITY of planar graphs stays NP-complete \cite{GareyJS76},
it is solvable in time $2^{O(\sqrt n)}$; see \cite{Marx13}.
Under the Exponential Time Hypothesis, PLANAR 3-COLORABILITY
cannot be solved in time $2^{o(\sqrt n)}$ (Marx \cite{Marx13}).
Recall that the consistency checking algorithm solves the problem 
in time $2^{O(k(n)\log n)}$
for any constructive upper bound $W(n)=O(k(n))$.
Therefore, the Exponential Time Hypothesis implies
the lower bound $\wplanar n=\Omega(\sqrt n/\log n)$.

Our goal is to estimate $\wplanar n$ unconditionally.
Note that the method we used to prove Theorem \ref{thm:lower}
does not apply to the planar case directly. For this approach
we would need planar Cai-Fürer-Immerman graphs.
Such graphs do not exist because every planar graph is
definable in $k$-variable logic with counting quantifiers
for some absolute constant $k$ \cite{Grohe98}, and for 3-connected planar graphs
this is true even without counting \cite{Verbitsky07}.
We also cannot use Theorem \ref{thm:girth}
because, by Grötzsch's theorem~\cite{Groetzsch58},
every 4-colorable planar graph has girth~3.
Nevertheless, we are able to show tight bounds for $\wplanar n$
combining Theorem \ref{thm:lower} with the standard reduction
of the general 3-COLORABILITY to its planar version.

\begin{theorem}\label{thm:planar}
  $\wplanar n=\Theta(\sqrt n)$.
\end{theorem}

  The upper bound $\wplanar n\le5\sqrt n$ immediately follows
from the general bound \refeq{tw} because it is known \cite{Grigoriev11} that,  
if $G$ is planar, then $\tw G\le5\sqrt n-1$.
In the rest of this section we prove a lower bound of~$\Omega(\sqrt n)$.

\begin{figure}
  \centering
 \begin{tikzpicture}
\begin{scope}[scale=.7]
\node at (-3cm,2.5cm) {(i)};
\path[every node/.style=uvertex] 
      (0,0) node (a00) {}
      (1,0) node (a10)  {} edge (a00)
      (2,0) node (a20)  {} edge (a10)
      (-1,0) node (a-10)  {} edge (a00)
      (-2,0) node (a-20)  {} edge (a-10)
      (0,1) node (a01) {} edge (a-10) edge (a00) edge (a10)
      (1,1) node (a11) {} edge  (a10) edge  (a20)
      (-1,1) node (a-11) {} edge  (a01) edge  (a-20)
     (0,-1) node (a0-1) {} edge (a-10) edge (a00) edge (a10)
      (1,-1) node (a1-1) {} edge  (a0-1) edge  (a20)
      (-1,-1) node (a-1-1) {} edge  (a-10) edge  (a-20)
      (0,2)  node (a02) {} edge (a-11) edge (a01) edge (a11)
      (0,-2)  node (a0-2) {} edge (a-1-1) edge (a0-1) edge (a1-1)
      (-3,0) node (u) {} edge[dashed] (a-20)
      (3,0) node (v) {} edge[dashed] (a20)
      (0,3) node (z) {} edge[dashed] (a02)
      (0,-3) node (w) {} edge[dashed] (a0-2);
   \node[left] (uu) at (u) {$u\,\,$};
   \node[right] (vv) at (v) {$\,\,v$};
   \node[left] (zz) at (z) {$z\,\,$};
   \node[left] (ww) at (w) {$w\,\,$};
\end{scope}

\begin{scope}[xshift=4.5cm]
\node at (0cm,1.8cm) {(ii)};
\path[every node/.style=vertex] 
      (0,0) node (u) {}
      (1,1) node (u+)  {} edge (u)
      (1,-1) node (u-)  {} edge (u)
      (2,0) node (a)  {} edge (u+) edge (u-) 
      (3,1) node (a+)  {} edge (a)
      (3,-1) node (a-)  {} edge (a)
      (4,0) node (b)  {} edge (a+) edge (a-) 
      (5,1) node (b+)  {} edge (b)
      (5,-1) node (b-)  {} edge (b)
      (6,0) node (c)  {} edge (b+) edge (b-) 
      (7,0) node (v)  {} edge (c)
      (1,2) node[draw=none,fill=none] (u++) {} edge[dashed] (u+)
      (1,-2) node[draw=none,fill=none] (u--) {} edge[dashed] (u-)
      (3,2) node[draw=none,fill=none] (a++) {} edge[dashed] (a+)
      (3,-2) node[draw=none,fill=none] (a--) {} edge[dashed] (a-)
      (5,2) node[draw=none,fill=none] (b++) {} edge[dashed] (b+)
      (5,-2) node[draw=none,fill=none] (b--) {} edge[dashed] (b-);
   \node[left] (uu) at (u) {$u\,\,$};
   \node[right] (vv) at (v) {$\,\,v$};
\end{scope}

\begin{scope}[scale=.7,yshift=-6.5cm]
\node at (-3cm,2cm) {(iii)};
\path
      (0,0) node[avertex] (a00) {}
      (1,0) node[bvertex] (a10)  {} edge (a00)
      (2,0) node[avertex] (a20)  {} edge (a10)
      (-1,0) node[bvertex] (a-10)  {} edge (a00)
      (-2,0) node[avertex] (a-20)  {} edge (a-10)
      (0,1) node[cvertex] (a01) {} edge (a-10) edge (a00) edge (a10)
      (1,1) node[cvertex] (a11) {} edge  (a10) edge  (a20)
      (-1,1) node[bvertex] (a-11) {} edge  (a01) edge  (a-20)
     (0,-1) node[cvertex] (a0-1) {} edge (a-10) edge (a00) edge (a10)
      (1,-1) node[bvertex] (a1-1) {} edge  (a0-1) edge  (a20)
      (-1,-1) node[cvertex] (a-1-1) {} edge  (a-10) edge  (a-20)
      (0,2)  node[avertex] (a02) {} edge (a-11) edge (a01) edge (a11)
      (0,-2)  node[avertex] (a0-2) {} edge (a-1-1) edge (a0-1) edge (a1-1);
\end{scope}

\begin{scope}[scale=.7,yshift=-6.5cm,xshift=7cm]
\path[every node/.style=uvertex] 
      (0,0) node[avertex] (a00) {}
      (1,0) node[bvertex] (a10)  {} edge (a00)
      (2,0) node[cvertex] (a20)  {} edge (a10)
      (-1,0) node[bvertex] (a-10)  {} edge (a00)
      (-2,0) node[cvertex] (a-20)  {} edge (a-10)
      (0,1) node[cvertex] (a01) {} edge (a-10) edge (a00) edge (a10)
      (1,1) node[avertex] (a11) {} edge  (a10) edge  (a20)
      (-1,1) node[avertex] (a-11) {} edge  (a01) edge  (a-20)
     (0,-1) node[cvertex] (a0-1) {} edge (a-10) edge (a00) edge (a10)
      (1,-1) node[avertex] (a1-1) {} edge  (a0-1) edge  (a20)
      (-1,-1) node[avertex] (a-1-1) {} edge  (a-10) edge  (a-20)
      (0,2)  node[bvertex] (a02) {} edge (a-11) edge (a01) edge (a11)
      (0,-2)  node[bvertex] (a0-2) {} edge (a-1-1) edge (a0-1) edge (a1-1);  
\end{scope}
 \end{tikzpicture}
  \caption{(i) The crossover gadget $C$ replacing the crossing of edges $uv$ and $wz$.
(ii) Along an edge $uv$, the crossover gadgets share their corner vertices.
One of the end vertices $u$ or $v$ (but not both) is identified with 
the corner vertex of the nearest gadget.
(iii) Up to permutation of colors, the crossover graph has exactly two
proper 3-colorings. }
  \label{fig:gadget}
\end{figure}
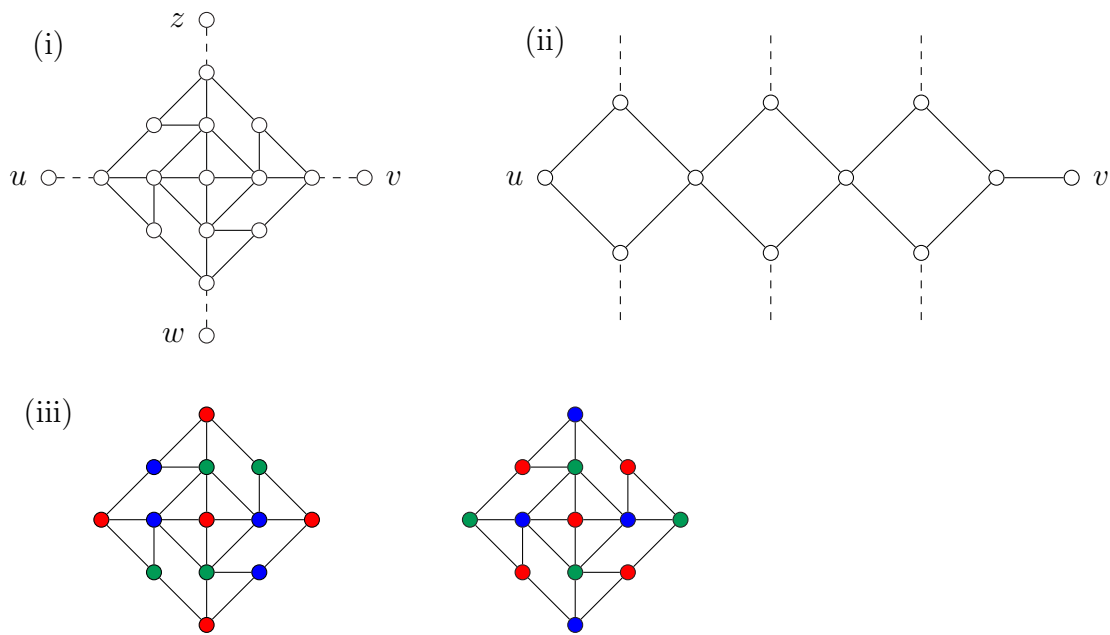
Recall the textbook reduction of 3-COLORABILITY to PLANAR 3-COLORA\-BI\-LITY \cite{GareyJS76}.
The reduction transforms an arbitrary graph $G$ into a planar graph $G'$ as follows.
First, a planar drawing $D$ of $G$ is made allowing edges crossings.
It is supposed that no more than two edges can cross at a point.
Then, each edge crossing in $D$ is replaced with the crossover gadget $C$ shown in Fig.~\ref{fig:gadget}.
The crossover gadget $C$ has two crucial properties:
\begin{enumerate}
\item[(a)] 
In any proper 3-coloring of $C$, the opposite corner vertices are equally colored;
\item[(b)]
any coloring of the four corner vertices where the opposite vertices
are equally colored uniquely extends to a proper 3-coloring of all of~$C$.
\end{enumerate}

Denote the set of 13 vertices in the gadget for the crossing edges $uv$ and $wz$ by $C_{uv,wz}$.
Thus, $C_{uv,wz}$ can contain one of $u$ and $v$ and/or one of $w$ and $z$, and
$V(G')=V(G)\cup\bigcup_{\{uv,wz\}}C_{uv,wz}$ where the union is over all edge crossings in $D$.
Property (a) implies the following fact.
\begin{enumerate}
\item[(A)] 
If $c'$ is a proper 3-coloring of $G'$ and $c$ is the restriction of $c'$ to $V(G)$,
then $c$ is a proper 3-coloring of~$G$.
\end{enumerate}
This property implies that if $\chi(G)>3$, then $\chi(G')>3$.

Given $X\subseteq V(G)$, define
\begin{equation}
  \label{eq:Xprime}
  X'=X\cup\bigcup C_{uv,wz}\text{ over adjacent $u,v\in X$ and all $wz$ crossing $uv$ in $D$}.
\end{equation}
The induced subgraph of $G$ spanned by $X$ will be denoted by $G[X]$.
We will exploit the following consequence of Property~(b).
\begin{enumerate}
\item[(B)] 
Let $X\subseteq V(G)$. Any proper 3-coloring $c$ of $G[X]$
admits an extension $c'$ to $X'$ that is a proper 3-coloring of $G'[X']$.
Moreover, all possible $c'$ are equal not only on $X$ but also on every $C_{uv,wz}$
such that $\{u,v,w,z\}\subseteq X$.
\end{enumerate}
Taking $X=V(G)$, we conclude that if $\chi(G)\le3$, then $\chi(G')\le3$.

\begin{lemma}\label{lem:reduction}
Let $\chi(G)>3$. If $W(G)>4k$ for an integer $k$, then $W(G')>k$.
\end{lemma}

Using Lemma \ref{lem:reduction}, we can now prove Theorem \ref{thm:planar}.
The construction in the proof of Theorem \ref{thm:lower} combined with 
Lemma \ref{lem:indist} gives us a non-3-colorable graph $G$ with
$n$ vertices, $e=O(n)$ edges, and $W(G)=\Omega(n)$. Let us
convert it into a non-3-colorable planar graph $G'$ as described above.
As the intermediate drawing $D$, we use a straight-line drawing of $G$
where edges are represented by segments of lines in general position;
hence, no three edges can share a crossing point.
Note that $D$ has less than $e^2=O(n^2)$ edge crossings.\footnote{%
This trivial bound for the number of edge crossings cannot be essentially
improved; it is known, for example, that most cubic graphs have
quadratic crossing number.}
Therefore, $G'$ has $N<n+13e^2=O(n^2)$ vertices.
By Lemma \ref{lem:reduction}, $W(G')=\Omega(n)=\Omega(\sqrt N)$.

Rigorously speaking, we have proved the bound $\wplanar N=\Omega(\sqrt N)$
for an infinite sequence of $N$. In order to get a desired graph for an
intermediate value of $N$, we construct $G'$ for the nearest number $N'<N$ 
in the sequence and pad it out with $N-N'$ isolated vertices.
To complete the proof, it remains to prove the lemma.

\begin{proofof}{Lemma \ref{lem:reduction}}
By Proposition \ref{prop:game}, it suffices to show that
Duplicator has a winning strategy for the $k$-width 3-coloring
game on $G'$. We will show that she can translate her
winning strategy for the $4k$-width 3-coloring
game on $G$ to the game on $G'$.
More precisely, the assumption $W(G)>4k$ implies that
Duplicator has a winning strategy in $\game{4k}rG$ for every $r$.
We will show that Duplicator's winning strategy for $\game{4k}{4r}G$
can be transformed to a winning strategy for $\game kr{G'}$.

Given $Y\subseteq V(G')$, we define the set $X=X(Y)\subseteq V(G)$ as follows:
$X$ contains all $Y\cap V(G)$, and each $y\in Y\setminus V(G)$ such that
$y\in C_{uv,wz}$ contributes the vertices $u$, $v$, $w$, and $z$ in $X$.
Note that $|X|\le4|Y|$ and $Y\subseteq X'$, where $X'$ is defined by~\refeq{Xprime}.
Property (B) of the reduction implies the following.

\begin{claim}\label{cl:}
  Let $Y\subseteq V(G')$ and $X=X(Y)$. Any proper 3-coloring $c$ of $G[X]$
extends to a proper 3-coloring $c'$ of $G'[X']$, and all such extensions
are equal on~$Y$.
\end{claim}

Let $Y_i$ denote the set of vertices colored after the $i$-th round of the game
$\game kr{G'}$. Duplicator's strategy we are going to describe will
have the following properties. 
\begin{enumerate}
\item[(P1)] 
$X_i=X(Y_i)$ appears as the set
of colored vertices after the $r_i$-th round of 
$\game{4k}{4r}G$, for some $r_i$ such that $i\le r_i\le4i$,
if in this game Duplicator uses a winning strategy and Spoiler uses the
strategy specially designed (simulated) depending on the strategy he follows in $\game kr{G'}$.
\item[(P2)]  
If $c_i$ is the coloring of $X_i$ in $\game{4k}{4r}G$
(which is a proper 3-coloring of $G[X_i]$ because Duplicator
follows a winning strategy) and $c'_i$ is its extension to
a proper 3-coloring of $G'[(X_i)']$ (existing by Claim \ref{cl:}), 
then the coloring $d_i$ of $Y_i$ in $\game kr{G'}$ is the restriction of 
$c'_i$ to~$Y_i$.
\end{enumerate}
A strategy ensuring these properties is obviously winning
because $d_i$ is a proper 3-coloring of $G'[Y_i]$.

We now have to define the simulated strategy for Spoiler in $\game{4k}{4r}G$.
It is essentially determined by the condition that the configuration
$X_i=X(Y_i)$ has to appear in $G$ once the configuration $Y_i$ appears in $G'$.
The coloring $c_i$ of $X_i$ is determined by Duplicator's strategy in $\game{4k}{4r}G$.
The coloring $d_i$ of $Y_i$ satisfying (P2) is completely determined by $c_i$ 
because, by Claim \ref{cl:}, all possible extensions $c'_i$ of $c_i$ coincide on~$Y_i$.
We will need to carefully check that $d_{i+1}$ agrees with $d_i$ on $Y_{i+1}\cap Y_i$;
only under this condition $d_i$ can be produced by some Duplicator's strategy in $\game kr{G'}$.

Consider the first round of $\game kr{G'}$.
Suppose that Spoiler asks Duplicator to color a vertex $y$;
thus, $Y_1=\{y\}$. If $y\in V(G)$, then
$X_1=\{y\}$. Otherwise $y$ belongs to some $C_{uv,wz}$,
and then $X_1=\{u,v,w,z\}$. Duplicator simulates $\game{4k}{4r}G$,
assuming that Spoiler claims the vertices of $X_1$ in the first rounds.
In accordance with Property (P2),
let $c_1$ be a proper 3-coloring of $G[X_1]$ according to Duplicator's
strategy for $\game{4k}{4r}G$ and $c'_1$ be an extension
of $c_1$ to a proper 3-coloring of $G'[(X_1)']$ according to Claim \ref{cl:}.
In the game $\game kr{G'}$, Duplicator assigns $y$ the color $d_1(y)=c'_1(y)$.

Assume now that Properties (P1) and (P2) are obeyed up to the 
$i$-th round and consider the $(i+1)$-th round of $\game kr{G'}$.

\Case 1{$|Y_i|<k$, $Y_{i+1}=Y_i\cup\{y\}$}
(that is, in the $(i+1)$-th round Spoiler asks to color a new vertex $y$).
Note that $X_i\subseteq X_{i+1}$, $|X_i|\le4(k-1)$, and $|X_{i+1}\setminus X_i|\le4$.
Duplicator simulates next rounds of $\game{4k}{4r}G$,
assuming that Spoiler claims the vertices in $X_{i+1}\setminus X_i$.
Let $c_i$ and $c_{i+1}$ be the colorings of $X_i$ and $X_{i+1}$
in $\game{4k}{4r}G$, $c_{i+1}$ being an extension of $c_i$.
By assumption, the coloring $d_i$ of $Y_i$ in $\game kr{G'}$
is obtainable by extending $c_i$ to a proper coloring $c'_i$
of $G'[(X_i)']$ and by restricting $c'_i$ to $Y_i$.
Similarly, extend $c_{i+1}$ to a proper coloring $c'_{i+1}$
of $G'[(X_{i+1})']$ and denote the restriction of $c'_{i+1}$
to $Y_{i+1}$ by $d_{i+1}$. Since $c_{i+1}$ extends $c_i$,
the restriction of $c'_{i+1}$ to $(X_i)'$ is an extension of $c_i$
to a proper coloring of $G'[(X_i)']$. By Claim \ref{cl:},
$c'_{i+1}$ and $c'_i$ coincide on $Y_i$. It follows that 
$d_{i+1}$ is an extension of $d_i$. Duplicator assigns $y$
the color $d_{i+1}(y)$, ensuring Properties (P1) and (P2)
also for the $(i+1)$-th round of $\game kr{G'}$.

\Case 2{$|Y_i|=k$, $Y_{i+1}=(Y_i\setminus\{y_1\})\cup\{y_2\}$}
(that is, Spoiler erases the color of $y_1$ and asks to color $y_2$ instead).
Let $Y_0=Y_i\setminus\{y_1\}$ and $X_0=X(Y_0)$. Note that $|X_0|\le4(k-1)$,
$|X_i\setminus X_0|\le4$, and $|X_{i+1}\setminus X_0|\le4$.
Duplicator simulates next rounds of $\game{4k}{4r}G$,
assuming that Spoiler erases the colors of the vertices in $X_i\setminus X_0$ 
and asks to color the vertices in $X_{i+1}\setminus X_0$.
Let $c_i$ and $c_{i+1}$ be the colorings of $X_i$ and $X_{i+1}$.
Note that they agree on $X_0$. Denote the (common) coloring of $X_0$
by $c_0$.
By assumption, the coloring $d_i$ of $Y_i$ in $\game kr{G'}$
is obtainable by extending $c_i$ to a proper coloring $c'_i$
of $G'[(X_i)']$ and by restricting $c'_i$ to $Y_i$.
Similarly, extend $c_{i+1}$ to a proper coloring $c'_{i+1}$
of $G'[(X_{i+1})']$ and denote the restriction of $c'_{i+1}$
to $Y_{i+1}$ by $d_{i+1}$. Since both $c_i$ and $c_{i+1}$ extend $c_0$,
the restrictions of $c'_i$ and $c'_{i+1}$ to $(X_0)'$ are extensions of $c_0$
to a proper coloring of $G'[(X_0)']$. By Claim \ref{cl:},
$c'_{i+1}$ and $c'_i$ coincide on $Y_0$. It follows that 
$d_{i+1}$ and $d_i$ coincide on $Y_0$ as well. Duplicator assigns the vertex $y_2$
the color $d_{i+1}(y_2)$, ensuring Properties (P1) and (P2)
for the $(i+1)$-th round of $\game kr{G'}$ also in this case.
\end{proofof}




\begin{thebibliography}{10}

\bibitem{Atserias04}
A.~Atserias.
\newblock On sufficient conditions for unsatisfiability of random formulas.
\newblock {\em J. ACM}, 51(2):281--311, 2004.

\bibitem{AtseriasBD07}
A.~Atserias, A.~A. Bulatov, and V.~Dalmau.
\newblock On the power of $k$-consistency.
\newblock In L.~Arge, C.~Cachin, T.~Jurdzinski, and A.~Tarlecki, editors, {\em
  ICALP 2007}, volume 4596 of {\em LNCS}, pages 279--290, Heidelberg, 2007.
  Springer.

\bibitem{AtseriasBD09}
A.~Atserias, A.~A. Bulatov, and A.~Dawar.
\newblock Affine systems of equations and counting infinitary logic.
\newblock {\em Theor. Comput. Sci.}, 410(18):1666--1683, 2009.

\bibitem{BartoK09}
L.~Barto and M.~Kozik.
\newblock Constraint satisfaction problems of bounded width.
\newblock In {\em Proceedings of the 50th Annual IEEE Symposium on Foundations
  of Computer Science}, pages 595--603. IEEE Computer Society, 2009.

\bibitem{BeigelE05}
R.~Beigel and D.~Eppstein.
\newblock 3-coloring in time {$O(1.3289^n)$}.
\newblock {\em J. Algorithms}, 54(2):168--204, 2005.

\bibitem{Berkholz13}
C.~Berkholz.
\newblock Lower bounds for existential pebble games and $k$-consistency tests.
\newblock {\em Logical Methods in Computer Science}, 9(4), 2013.

\bibitem{BerkholzV13}
C.~Berkholz and O.~Verbitsky.
\newblock On the speed of constraint propagation and the time complexity of arc
  consistency testing.
\newblock In K.~Chatterjee and J.~Sgall, editors, {\em Mathematical Foundations
  of Computer Science 2013 - 38th International Symposium, MFCS 2013,
  Klosterneuburg, Austria, August 26-30, 2013. Proceedings}, volume 8087 of
  {\em Lecture Notes in Computer Science}, pages 159--170. Springer, 2013.

\bibitem{Bessiere06}
C.~Bessi\`{e}re.
\newblock {\em Handbook of Constraint Programming}, chapter Constraint
  Propagation, pages 29--84.
\newblock Elsevier, Amsterdam, 2006.

\bibitem{BulatovKL08}
A.~A. Bulatov, A.~A. Krokhin, and B.~Larose.
\newblock Dualities for constraint satisfaction problems.
\newblock In N.~Creignou, P.~G. Kolaitis, and H.~Vollmer, editors, {\em
  Complexity of Constraints}, volume 5250 of {\em Lecture Notes in Computer
  Science}, pages 93--124. Springer, 2008.

\bibitem{CaiFI92}
J.-Y. Cai, M.~F{\"u}rer, and N.~Immerman.
\newblock An optimal lower bound on the number of variables for graph
  identifications.
\newblock {\em Combinatorica}, 12(4):389--410, 1992.

\bibitem{Dawar98}
A.~Dawar.
\newblock A restricted second order logic for finite structures.
\newblock {\em Inf. Comput.}, 143(2):154--174, 1998.

\bibitem{Erdos59}
P.~{Erd\H{o}s}.
\newblock {Graph theory and probability.}
\newblock {\em {Can. J. Math.}}, 11:34--38, 1959.

\bibitem{Erdos62}
P.~{Erd\H{o}s}.
\newblock {On circuits and subgraphs of chromatic graphs.}
\newblock {\em {Mathematika}}, 9:170--175, 1962.

\bibitem{FederV98}
T.~Feder and M.~Y. Vardi.
\newblock The computational structure of monotone monadic {SNP} and constraint
  satisfaction: A study through datalog and group theory.
\newblock {\em SIAM J. Comput.}, 28(1):57--104, 1998.

\bibitem{Freuder90}
E.~C. Freuder.
\newblock Complexity of $k$-tree structured constraint satisfaction problems.
\newblock In H.~E. Shrobe, T.~G. Dietterich, and W.~R. Swartout, editors, {\em
  Proceedings of the 8th National Conference on Artificial Intelligence}, pages
  4--9. AAAI Press / The MIT Press, 1990.

\bibitem{GareyJS76}
M.~R. Garey, D.~S. Johnson, and L.~J. Stockmeyer.
\newblock
Some simplified NP-complete graph problems.
\newblock
{\em Theor. Comput. Sci.}, 1(3):237-267, 1976.

\bibitem{GreenwellL74}
D.~{Greenwell} and L.~{Lov\'asz}.
\newblock Applications of product colouring.
\newblock {\em Acta Math. Acad. Sci. Hung.}, 25:335--340, 1974.

\bibitem{Grigoriev11}
A.~Grigoriev.
\newblock Tree-width and large grid minors in planar graphs.
\newblock {\em Discrete Mathematics {\&} Theoretical Computer Science},
  13(1):13--20, 2011.

\bibitem{Grohe98}
M.~Grohe.
\newblock Fixed-point logics on planar graphs.
\newblock In {\em 13-th Annual IEEE Symposium on Logic in Computer Science},
  pages 6--15. IEEE Computer Society, 1998.

\bibitem{Groetzsch58}
H.~Grötzsch.
\newblock Zur {T}heorie der diskreten {G}ebilde. {VII}. {E}in {D}reifarbensatz
  für dreikreisfreie {N}etze auf der {K}ugel.
\newblock {\em Wiss. Z. Martin-Luther-Univ. Math.-Nat. Reihe}, 8:109--120,
  1958/1959.

\bibitem{HellN90}
P.~Hell and J.~Nesetril.
\newblock On the complexity of {$H$}-coloring.
\newblock {\em J. Comb. Theory, Ser. B}, 48(1):92--110, 1990.

\bibitem{HellNZ96}
P.~Hell, J.~Nesetril, and X.~Zhu.
\newblock Complexity of tree homomorphisms.
\newblock {\em Discrete Applied Mathematics}, 70(1):23--36, 1996.

\bibitem{ImpagliazzoPZ01}
R.~Impagliazzo, R.~Paturi, and F.~Zane.
\newblock Which problems have strongly exponential complexity?
\newblock {\em J. Comput. Syst. Sci.}, 63(4):512--530, 2001.

\bibitem{KolaitisV95}
P.~G. Kolaitis and M.~Y. Vardi.
\newblock On the expressive power of datalog: Tools and a case study.
\newblock {\em J. Comput. Syst. Sci.}, 51(1):110--134, 1995.

\bibitem{Marx13}
D.~Marx.
\newblock The square root phenomenon in planar graphs.
\newblock In F.~V. Fomin, R.~Freivalds, M.~Z. Kwiatkowska, and D.~Peleg,
  editors, {\em Automata, Languages, and Programming - 40th International
  Colloquium, ICALP 2013, Riga, Latvia, July 8-12, 2013, Proceedings, Part II},
  volume 7966 of {\em Lecture Notes in Computer Science}, page~28. Springer,
  2013.

\bibitem{NesetrilZ96}
J.~Ne\v{s}et\v{r}il and X.~Zhu.
\newblock On bounded treewidth duality of graphs.
\newblock {\em Journal of Graph Theory}, 23(2):151--162, 1996.

\bibitem{Pinsker73}
M.~S. Pinsker.
\newblock On the complexity of a concentrator.
\newblock In {\em 7-th International Teletraffic Conference}, pages
  318/1--318/4, 1973.

\bibitem{ShiW07}
L.~Shi and N.~C. Wormald.
\newblock Colouring random regular graphs.
\newblock {\em Combinatorics, Probability {\&} Computing}, 16(3):459--494,
  2007.

\bibitem{Verbitsky07}
O.~Verbitsky.
\newblock Planar graphs: {L}ogical complexity and parallel isomorphism tests.
\newblock In W.~Thomas and P.~Weil, editors, {\em 24th Annual Symposium on
  Theoretical Aspects of Computer Science}, volume 4393 of {\em Lecture Notes
  in Computer Science}, pages 682--693. Springer, 2007.

\end{thebibliography}

\end{document}